\documentclass[12pt]{article}

\usepackage[utf8]{inputenc}

\usepackage{amsmath,amssymb,amsthm,mathtools,amstext}

\usepackage{graphicx}
\usepackage{subfig}
\usepackage{subcaption}
\usepackage{tikz}
\usetikzlibrary{shapes.geometric, arrows}
\usepackage{pgfplots}
\pgfplotsset{compat=1.17}
\usepackage{pgf-pie}
\usepackage{svg}

\usepackage{xcolor}
\usepackage{titlesec}

\usepackage{ulem}
\usepackage[round,authoryear]{natbib}
\usepackage{verbatim}
\usepackage{chapterbib}
\usepackage{hyperref}
\usepackage{cancel}
\usepackage{multirow}
\usepackage{multicol}
\usepackage{eurosym}
\usepackage{academicons}
\usepackage[left]{lineno}  

\newtheorem{proposition}{Proposition}[section]
\newtheorem{corollary}[proposition]{Corollary}
\newtheorem{theorem}{Theorem}
\newtheorem{assumption}[proposition]{Assumption}
\newtheorem{lemma}[proposition]{Lemma}
\theoremstyle{definition}
\newtheorem{definition}[proposition]{Definition}
\newtheorem{remark}[proposition]{Remark}
\newtheorem{notation}[proposition]{Notation}

\begin{document}

\title{\LARGE Consistency and Central Limit Results for the Maximum Likelihood Estimator in the Admixture Model}

\maketitle

\begin{center}
   {Carola Sophia Heinzel} 

Department of Mathematical Stochastics,\\
Ernst-Zermelo Straße 1,
            Freiburg im Breisgau,
            79140,Germany \\
            carola.heinzel@stochastik.uni-freiburg.de
\end{center}

\begin{abstract}
In the Admixture Model, the probability of an individual having a certain number of alleles at a specific marker depends on the allele frequencies in $K$ ancestral populations and the fraction of the individual's genome originating from these ancestral populations.

This study investigates consistency and central limit results of maximum likelihood estimators (MLEs) for the ancestry and the allele frequencies in the Admixture Model, complimenting previous work by \cite{pfaff2004information, pfaffelhuber2022central}. Specifically, we prove consistency of the MLE, if we estimate the allele frequencies and the ancestries. Furthermore, we prove central limit theorems, if we estimate the ancestry of a finite number of individuals and the allele frequencies of finitely many markers, also addressing the case where the true ancestry lies on the boundary of the parameter space. 
Finally, we use the new theory to quantify the uncertainty of the MLEs for the data of \citet{10002015global}.
\end{abstract}


\section{Introduction}
Estimating the individual's ancestry from genetic data has many applications, e.g., identifying missing persons \citep{phillips2015}, exploring human history \citep{rb2002}, or correcting for population stratification \citep{pritchard2001}. Therefore, the Admixture Model is often used, see e.g., \citep{lovell2021, badouin2017} for examples of application. More precisely, we have genetic data that are multinomially distributed with a success probability that depends on the individual's ancestries in $K$ ancestral populations and the allele frequencies in these populations. We estimate the individuals' ancestries and allele frequencies with maximum likelihood estimators (MLEs). MLEs in the Admixture Model occur in two different settings: the allele frequencies are known (supervised setting) or the allele frequencies must be estimated (unsupervised setting). Here, we introduce a third setting, the semi-supervised setting. In this framework, we assume complete knowledge of ancestry for infinitely many individuals and of allele frequencies across infinitely many genetic markers. Nevertheless, there are $N_C$ individuals with unknown ancestry and $M_C$ markers with unknown allele frequencies.

Uncertainty of MLEs can have, besides of having only a limited amount of data,  different reasons. This includes, e.g., the case where the model does not fit the data \citep{lawson2018}. To detect this, the correlation between the true genotypes and the expected genotypes given the predicted allele frequencies can be used \citep{garcia2020evaluation}. Furthermore, \citet{divers2011} considered the influence of measurement errors on the quality of the estimation. Here, we exclusively focus on the Admixture Model and assume that neither measurement errors occur nor the model does not fit the data. 

There are some studies that have considered the uncertainty in the Admixture Model under these assumptions. For example, \citet{pfaff2004information} named a central limit theorem (CLT) in the supervised setting and \citet{pfaffelhuber2022central} calculated the asymptotic distribution of the MLE for normally distributed allele frequencies, if the true parameters are in the interior of the parameter space. We complement existing methods by the CLT if the true parameter is on the boundary of the parameter space and a CLT for the semi-supervised setting, for which we need consistency of the MLEs. Especially, uniqueness of the MLEs is important. Hence, we name constraints that ensure the uniqueness of the MLEs in the unsupervised setting which is based on \cite{heinzel2025revealing}.

Therefore, we use some general theory on consistency and central limit results for MLEs. For independently and identically distributed (i.i.d.) random variables, this is textbook knowledge \citep{ferguson2017}. For independently, not identically distributed random variables, \citet{hoadley1971} developed some theory. Moreover, \cite{redner1981note} considered consistency for non-identifiable distributions of the random variables, which are assumed to be i.i.d.. Unfortunately, both results are not directly applicable to the Admixture Model, i.e we adapted the proof in \cite{hoadley1971}. If the true parameters are on the boundary of the parameter space, we use theory by \citet{andrews1999estimation}. 

First, we explain the model in more detail and name already known results that we will use later (section \ref{sec:model}). In section \ref{sec:main} we explain our main results, i.e.  consistency and central limit results for different settings. Afterwards, we apply the CLT in the Admixture Model to real data from \citet{10002015global} to quantify the uncertainty of the MLEs. The results, which are presented in section \ref{sec:application}, show that the uncertainty for the estimation is much smaller, if the true parameter is on the boundary of the parameter space. Finally, we prove our claims in section \ref{sec:proof}.

\section{Admixture Model}\label{sec:model}
In this section, we introduce the most important notation and we define the Admixture Model as e.g. described by \cite{alexander2009fast}.  

\begin{definition}[Admixture Model]
Let $M,N \in \mathbb N, i \in \{1, ..., N\}, n \in \{1, ..., M\}$. Let $\mathbb{S}^K$ for the $(K-1)$-dimensional simplex. Moreover, we denote by $q_{i, \cdot}$ the row vector $(q_{i,1}, \ldots, q_{i,K  })$ and we write $\langle v_1, v_2 \rangle$ for the scalar product of the vectors $v_1, v_2$. Let $K$ be fixed, $q^0_{i, \cdot} \in \mathbb{S}^K$, and $p^0_{k,\cdot,m}\in \mathbb{S}^{J_m}.$ Additionally, it holds $X_{i, \cdot, m} \sim \text{Multi}(2, \langle q^0_{i, \cdot} p^0_{\cdot, j, m}\rangle_{j=1,...,{J_m}})$, and the family $(X_{i, \cdot, m})_{i=1,...,N, m=1,...,M}$ is independent. Hence, the log-likelihood is, for a constant $C_x$ which only depends on $x$,
\begin{align*}
    \ell(q^0, p^0|x)  & = C_x +  \frac{1}{MNJ}\sum_{i=1}^N \sum_{m=1}^M \sum_{j=1}^{J_m} x_{i,j,m} \log( \underbrace{\langle q^0_{i, \cdot}, p^0_{\cdot, j, m} \rangle}_{=:c_{i, j,m}^0}).
  \end{align*}
The MLE is defined by
\begin{align*}
         \left(\hat Q^N, \hat P^M\right) &= \text{argmax}\left\{(q,p) \mapsto \ell(q,p|x)\right\}.
\end{align*}
\end{definition}

The interpretation is that we have data $x =(x_{i,j,m})_{i=1,...,N, j=1,...,J, m=1,...,M}$, where $x_{i,j,m}$ stands for the number of copies of allele $j$ in individual $i$ at locus $m$. The frequencies of allele $j$ at marker $m$ in population $k$ is denoted by $p_{k,j,m},$ while the ancestry of individual $i$ from population $k$ is called $q_{i,k}.$

\begin{remark}
      We assume bi-allelic markers, i.e.  $J_m = 2 \,\forall m \in \{1, \ldots, M\}$. Hence, we write $p_{\cdot, m}$ for the frequency of one of the alleles at marker $m$ and $c_{i,m}^0 := \langle q^0_{i, \cdot} p^0_{\cdot,m}\rangle$. This simplifies the proofs, but the proof ideas can adapted to the general case. Moreover, we define
\begin{align*}
    \ell(q,p|X_{i,m} = x_{i,m}) := \log(\mathbb P_{q,p}(X_{i,m} = x_{i,m})).
\end{align*}
\end{remark}

As e.g. shown in \cite{heinzel2025revealing}, the MLE is non-unique in the unsupervised setting. For an invertible, $K \times K$ matrix $S_K,$ it holds
\begin{align}
    \mathbb P\left(\textup{Bin}(2, \langle q_{i, \cdot},p_{\cdot,m}\rangle) = x_{i,m}\right) = \mathbb P\left(\textup{Bin}(2, \langle q_{i, \cdot}S_K,S^{-1}_Kp_{\cdot,m}\rangle) = x_{i,m}\right).  \label{eq:nu}
\end{align}

Conditions that ensure that $q_{i, \cdot} S_K \geq 0$ can be interpreted as the ancestry of individual $i$ and that $S_K^{-1} p_{\cdot,  m}$ can be interpreted as the allele frequencies of marker $m$ are named in Section \ref{sec:uniqueness}. 

An other reason for the non-uniqueness might be Label Switching \citep{jakobsson2007, alexander2009fast, pritchard2010, verdu2014, kopelman2015, behr2016, cabreros2019likelihood, liu2024clumppling}, which is defined as follows.

\begin{definition}[Label Switching]\label{def:ls}
Let $\left(\left(q^0\right)^{\mathrm{T}}, p^0\right)$ be fixed. Then, we say that a permutation of the rows of this matrix is \textit{label switching.}
\end{definition}
Label Switching means that we choose a permutation of the unit matrix for $S_K$. 

\section{Main Results}\label{sec:main}

In this section, the main results of this study such as consistency and CLTs for different cases under weak and usually necessary assumptions are outlined. Therefore, we always write
 $\mathbb E,$ for the expected value with respect to the true values (usually $\left(q^0, p^0\right)$). The proofs can be found in Section \ref{sec:proof}.

The following theorem states that the MLE tends towards a value that has almost the same likelihood as the true parameters $(q^0, p^0)$.

\begin{theorem}[Consistency for non-unique MLEs]\label{th:wc_us}
  Let $\Theta$ be the parameter space. We define
    \begin{align*}
d\left(\left(q^N,p^M\right), \left(q^0,p^0\right)\right) &:= \sum_{i=1}^{N}\frac{|q^N_{i, \cdot} - q^0_{i, \cdot}|}{2^i} + \sum_{m=1}^{M} \frac{|p^M_{\cdot, m} - p^0_{\cdot,m}|}{2^m}, \\
\mathcal M(\gamma)&:= \left\{(q,p) \in \Theta: \lim_{M,N \to \infty} \frac{1}{MN} \sum_{m=1}^M \sum_{i=1}^N |c_{i,m} - c_{i,m}^0| \geq \gamma \right\}.
    \end{align*}
    Then, it holds for any $\epsilon, \gamma > 0$
    \begin{align*}
        \lim_{M,N \to \infty} \mathbb P\left(\min\left\{d\left(\left(\hat Q^N, \hat P^M\right), (q,p)\right): (q,p) \in \mathcal M(\gamma)\right\} \geq \epsilon \right) = 0.
    \end{align*}
\end{theorem}

In the sequel, we write $\mathcal M := \mathcal M(\gamma)$ for every $\gamma > 0.$
 We aim to define $\mathcal M$ as small as possible, i.e. the set $ \{(q,p) \in \Theta: \ell(q,p|X) = \ell\left(q^0, p^0|X\right)\}$ would be ideal.
Note that the set $\mathcal M$ never contains only one element. This is caused by the normalization of the log-likelihood function, the non-uniqueness of the MLE as highlighted in \eqref{eq:nu} and by label switching. 

We are not only interested in consistency of the MLEs, but also in CLTs for the MLEs, if either the parameter space is open or the true parameter is on the boundary of the parameter space. In this case, we estimate the allele frequencies of $M_C$ markers and the ancestries of $N_C$ individuals, the MLEs for infinitely many markers and individuals is not unique, i.e. we cannot hope for asymptotic normality in this case. Therefore, we first need some notation.

\begin{notation}
We define
 $$\Theta_o := \left\{(q,p) \in \left((0,1)^{N_C \times (K-1)}, (0,1)^{K \times M_C}\right): \sum_{k=1}^{K-1} q_{i,k} < 1 \, \forall i \in \{1, ..., N_C\}\right\}$$
and $(q^0, p^0) \in \Theta_o$ and $\tilde q^0, \tilde p^0$ for the true ancestries of all individuals and allele frequencies of all markers. If the limits exist, we define
 \begin{align*}
   \bar \Gamma\left(q_{i, \cdot}'\right) &:= \lim_{M \to \infty} \frac{-1}{M}\sum_{m = 1}^M \mathbb E\left(\frac{\partial^2 \ell(q,p|X_{i,m})}{\partial (q_{i, 1},..., q_{i,K-1})^2}\Bigg|_{q_{i, \cdot} = q_{i, \cdot}'}\right),\\
       \bar{\Gamma}_{(q,q), i}(q', p') &:= \lim_{M \to \infty} \frac{-1}{M}\sum_{m=1}^M \mathbb E\left(\frac{\partial^2 \ell(q,p|X_{i,m})}{\partial (q_{i, 1},..., q_{i,K-1})^2}\bigg|_{(q,p) = (q', p')}\right) ,\\
 \bar{\Gamma}_{(p,p),m}(q', p')&:= \lim_{N \to \infty} \frac{-1}{N} \sum_{i=1}^N \mathbb E\left(\frac{\partial^2  \ell(q,p|X_{i, m})}{ \partial (p_{1,m},..., p_{K,m})^2}\bigg|_{(q,p) = (q', p')}\right) , \\
    \bar{\Gamma}_{var(q)}(q, p) &:= \begin{pmatrix}
        \bar{\Gamma}_{(q,q),1}(q, p) & 0& ... & 0 \\
         0 & \bar{\Gamma}_{(q,q),2}(q, p) & ... & 0 \\
        \vdots & \ldots  & \ddots  & \vdots\\
         0 & 0 & ... & \bar{\Gamma}_{(q,q),N_C}(q, p)  \\
    \end{pmatrix}.
\end{align*}
Analogue to $\bar{\Gamma}_{var(q)}(q, p)$, we define $\bar{\Gamma}_{var(p)}(q, p)$ by $\left(\bar{\Gamma}_{(p,p),m}(q, p)\right)_{m = 1,..., M_C}$. 
Finally, we define  $$\bar \Gamma\left(q, p\right) := \begin{pmatrix}
        \bar{\Gamma}_{var(q)}( q, p) & 0 \\
        0 & \bar{\Gamma}_{var(p)}( q,  p)
    \end{pmatrix} \in \mathbb R^{(N_C (K-1) + M_C K) \times (N_C(K-1) + M_C K)}.$$
    We define 
\begin{align*}
            \underline{p_{Km}} &:= (p_{K,m},..., p_{K,m}) \in (0,1)^{K-1}, \\
        V_1 &:= \left\{\left(\tilde p^0_{1, m}, \tilde p^0_{2,m}, \ldots, \tilde p^0_{K-1,m}\right) -  \underline{\tilde p^0_{Km}}: m= 1,2, \ldots\right\},\\
        V_2 &:= \{\tilde q^0_{i, \cdot}: i=1, 2, \ldots\},\\
        \underline{p_{s,m}} &:= ((p_{1,m}, p_{2,m},..., p_{K-1,m}) - \underline{p_{Km}})^\top ((p_{1,m}, p_{2,m},..., p_{K-1,m}) - \underline{p_{Km}}).
    \end{align*} 
\end{notation}

We first name a CLT for an open parameter space, for which we need Assumption \ref{ass:invert} to ensure the uniqueness of the MLE and the existence of the Fisher information. 

\begin{assumption}\label{ass:invert}
   Without loss of generality, we assume that we estimate the ancestries of individuals $1, ..., N_C$ and the allele frequencies of markers $1,..., M_C.$ Additionally, we assume that $q_{i, \cdot}, p_{\cdot, m}$ are known for $i > N_C, m > M_C$. We assume that it holds $\tilde q^0_{i, k} \in [\epsilon_q, 1-\epsilon_q]$ for some $\epsilon_q > 0$ and that $\tilde p^0_{i, k} \in [\epsilon_p, 1-\epsilon_p]$ for some $\epsilon_p > 0$. 
We assume that
\begin{itemize}
    \item[(*)]  there exist infinitely many disjoint subsets \( S^p_1, S^p_2, \dots \subset V_1\), where each subset \( S^p_i \) consists of $K-1$ linearly independent vectors.
    \item[(**)] there exist infinitely many disjoint subsets \( S^q_1, S^q_2, \dots \subset V_2\), where each subset \( S^q_i \) consists of $K$ linearly independent vectors.
    \item[(***)] it holds
 \begin{align*}
    \lim_{N \to \infty} \frac{1}{N}\sum_{i=1}^N \frac{2}{\left(\tilde c_{i,m}^0(1-\tilde c_{i,m}^0)\right)^3}1_{\tilde c_{i,m}^0 \in(0,1)}  < \infty \, \forall m \in \{1,..., M_C\}, \\
       \lim_{M \to \infty} \frac{1}{M} \sum_{m=1}^M \frac{ 2}{\left(\tilde c_{i,m}^0(1-\tilde c_{i,m}^0)\right)^3}1_{\tilde c_{i,m}^0 \in(0,1)} < \infty \, \forall i \in \{1,..., N_C\}.
\end{align*} 
\end{itemize}
Additionally, we assume that the MLE is unique. 
\end{assumption}

The assumption that the MLE is unique seems to be strong at the first glance. However, we name easy to verify conditions, that ensures the uniqueness of the MLEs in Section \ref{sec:uniqueness}. We finally name a CLT in the semi-supervised setting.

\begin{theorem}[Central Limit Theorem for an open parameter space]\label{th:clt_unsupervised}
Let Assumption \ref{ass:invert} hold. Moreover, let $\bar \Gamma\left(\tilde q^0, \tilde p^0\right) \succ 0$ for all $(\tilde q_{1:N_C, \cdot}, \tilde p_{\cdot, 1:M_C}) \in \Theta_o.$
Then, it holds 
$$\left(\sqrt{M} \left(\hat Q^{N_C} - q^0\right),\sqrt{N} \left(\hat P^{M_C} -p^0\right)\right) \xRightarrow{N, M \to \infty}\mathcal N\left(0, \bar \Gamma^{-1}\left(\tilde q^0, \tilde p^0\right)\right).$$
\end{theorem}

Consideration of a closed parameter space is in particular important, as this ensures the uniqueness of the MLE (see section \ref{sec:uniqueness}). Moreover, ADMIXTURE often names estimators that are on the boundary of the parameter space. Hence, we also name a central limit theorem for this case.

\begin{theorem}[Central Limit Theorem for a Closed Parameter Space]\label{th:uv_closed}
 Let (*) hold and let $N = 1$. We write $q^0_{k}, c^0_m, X_m$ instead of $q^0_{1, k}, c^0_{1,m}, X_{1,m}$. We assume without loss of generality $q_K^0 > 0.$
There exists an $\epsilon > 0$ such that $c^0_m \in [0, \epsilon] \cup [1-\epsilon, 1]$ for maximal finitely many $m \in \{1, ..., M\}$ and let $\mathcal K_{min}:= \{k \in \{1,...,K-1\}: q_k^0 = 0\}, \mathcal K_{max}:= \{k \in \{1,...,K-1\}: q_k^0 = 1\}, Z \sim \mathcal N\left(0, \bar \Gamma^{-1}(q^0) \right), \Gamma\left(q^0\right) \succ 0$. We define
\begin{align*}
     \Lambda &:= \left\{v \in \mathbb R^{K-1}: v_j \leq 0 (j \in \mathcal K_{max}), v_i \geq 0 (i \in \mathcal K_{min})\right\}, \\
    \hat \lambda &:= \textup{argmin}\left\{\lambda \mapsto (\lambda - Z)^{\top} \bar \Gamma\left(q^0\right)(\lambda - Z) : \lambda \in \Lambda \right\}.
\end{align*}
Then, it holds
\begin{align*}
  \sqrt{M}\left(\hat Q^1_{\cdot}-q^0_{\cdot}\right) \xRightarrow{M \to \infty}\hat \lambda.
\end{align*}
\end{theorem}

It is possible to calculate $\hat \lambda$ precisely. For $K = 2, 3$ this calculation can be found on \href{https://github.com/CarolaHeinzel/Limit_Results_Admixture_Model}{GitHub}.
\begin{remark}
The assumption $q_K^0 > 0$ has the following reason: Let $q^0 = (q_1, 1-q_1, 0).$ Then, the set $\Lambda$ does not take into account that $\hat Q^{1}_{K} -q^0_K \geq 0$ has to hold. However, the assumption does not restrict our model.
\end{remark}

\section{Application to Data}\label{sec:application}
We apply the central limit results to data from  \citet{10002015global} to infer the uncertainty of the MLEs. To calculate the MLEs, we ran ADMIXTURE \citep{alexander2009fast} for all 2504 individuals and $K = 3$. The markers from the Kidd AIM set \citep{kidd2014} that consists of $M=55$ bi-allelic markers have been used. The marginal densities for one individual are shown in Figure \ref{fig:1}. 

\begin{figure}[ht]
    \centering
    \begin{minipage}{0.48\textwidth}
        \centering
        \includegraphics[width=\linewidth]{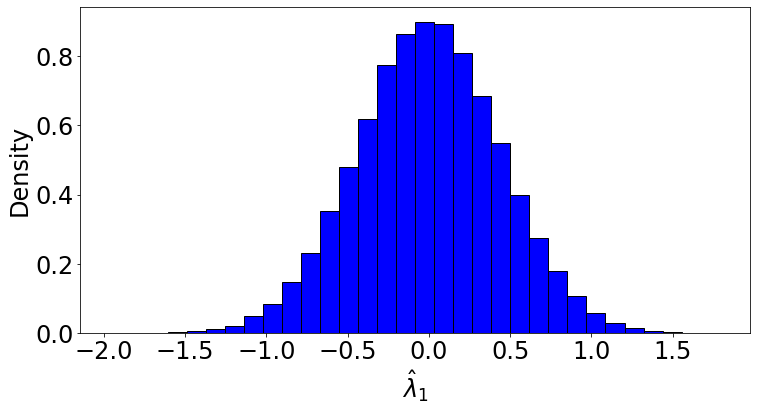}
    \end{minipage}
    \hfill
    \begin{minipage}{0.48\textwidth}
        \centering
        \includegraphics[width=\linewidth]{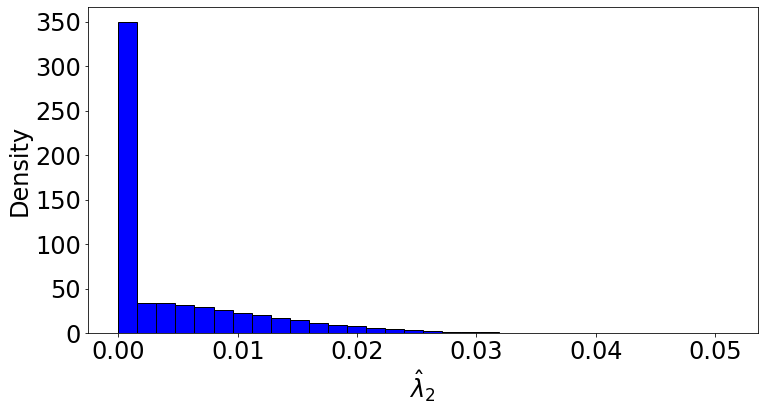}
    \end{minipage}
  \caption{\small{Asymptotic Distribution of the MLE $\sqrt{M}\left(\hat Q^1 -q^0\right)$ for individual HG00096 from GBR with $M = 55$. The output of ADMIXTURE for this individual was $(0.937166, 0.000010, 0.062824)$.}}\label{fig:1}
\end{figure}

The density is normally distributed for $\hat \lambda_1,$ where the MLE was in the interior of the parameter space. The MLE for the second population is on the boundary of the parameter space, which explains the non-normality of the marginal density.

Additionally, the results for individual HG00096 with $K = 2$ are shown in Figure \ref{fig:2}. Here, we showed 
$$\hat \lambda_2 = \begin{cases}
    0, \textup{ with probability 0.5} \\
    \tilde \lambda_2:= (\mathcal N(0, \bar{\Gamma}(q^0)) \lor 0), \textup{ else.}
\end{cases}$$

\begin{figure}[ht]
    \centering
    \includegraphics[width=0.5\linewidth]{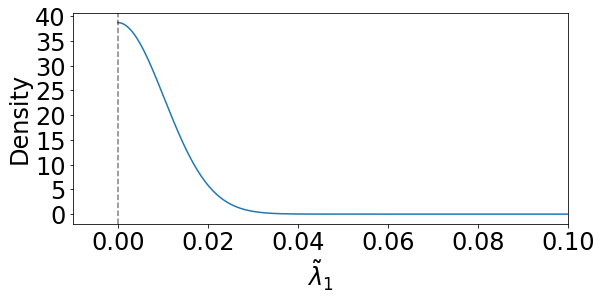}
    \caption{\small{Asymptotic Distribution of the MLE $\sqrt{M}\left(\hat Q^1 -q^0\right)$ for individual HG00096 from GBR with $M = 55$. The output of ADMIXTURE for this individual was $(0.000010, 0.999990)$.}}
    \label{fig:2}
\end{figure}

Figures \ref{fig:1} and \ref{fig:2} demonstrate that reducing the number of ancestral populations from $K=3$ to $K=2$ does not resolve the issue of the MLE lying on the boundary of the parameter space. Moreover, we see that the variance of the MLE is much smaller than if it is in the interior of the parameter space. The intuitive reason for this is that the difference between the MLE and the true parameter can only be either positive or negative and not both as in the case in which the true parameter is in the interior of the parameter space.

\section{Proofs of the Main Results}\label{sec:proof}
We proceed by proving the main results of this study. Therefore, we first name the conditions that ensure the uniqueness of the MLE. Afterwards, we prove consistency and the central limit results.

\subsection{Uniqueness of the MLE}\label{sec:uniqueness}

We consider the non-uniqueness of the MLE in the unsupervised setting for finitely many markers and individuals, i.e. we we aim to find given one MLE every matrix $S_K$ from equation \eqref{eq:nu}, where we can interpret $S^{-1}_K \hat p_{\cdot,m}$ as the allele frequencies and $\hat q_{i, \cdot}S_K$ as the individuals' ancestries. These matrices are called \textit{possible}, which is defined more precisely in Definition \ref{def:possible}.

\begin{definition}[Possible Matrix $S_K$]\label{def:possible}
 We say that the matrix $S_K$ is \textit{possible} (for the estimators $\hat q, \hat p$), if $S_K$ is invertible and it holds
\begin{align}
    (\hat qS_K)_{i, \cdot} \in \mathbb S^K, (S_K^{-1} \hat p)_{k, \cdot, m} \in \mathbb S^2. \label{constraint}
\end{align}   
Additionally, we define $\mathcal P(E_K)$ as the set that contains the $K$ dimensional unit matrix and all permutations of it.
\end{definition}

Corollary \ref{lemma:unique} states that the only possible matrices for $K = 2$ are the unit matrix and the permutation of the unit matrix if and only if specific conditions hold.

\begin{corollary}[Unique MLE for $K = 2$]\label{lemma:unique}
Let $K = 2$. It holds $(\exists m_1, m_2: \hat p_{1,m_1} = 1, \hat p_{1,m_2} = 0 ,\hat p_{2, m_1} \neq \hat p_{1,m_1}, \hat p_{1,m_2} \neq \hat p_{2,m_2}$, $ \exists i_1, i_2: \hat q_{i_1, 1} = 1, \hat q_{i_2, 1} = 0)$ if and only if ($S_2$ fulfills \eqref{constraint} $\Rightarrow S_2 \in \mathcal P(E_2)).$
\end{corollary}
\begin{proof}
Since the values of $N, M$ can be ignored here, we write $\hat q_{i,1}, \hat p_{2,j,m}$ instead of $\hat q_{i,1}^N, \hat p_{2,j,m}^M$. Let 
$$S_2 := \begin{pmatrix}
1 - a & a \\
b & 1- b\\
\end{pmatrix}.$$
Label switching for $K = 2$ occurs if and only if $a+b < 1$ and $a+b > 1$ are both possible \citep{heinzel2025revealing}. This is the reason why we assume, without loss of generality, $a+b<1$.
We have to choose $a,b$ that so that \begin{align}
    (qS_2)_{i,\cdot} &\in \mathbb S^2, \label{c_q}\\
    (S_2^{-1})_{k,\cdot,m} &\in \mathbb S^2. \label{c_p}
\end{align}
We define
\begin{align*}
        \hat q^{\min} &\coloneqq \min\{\hat q_{i,1}:i \in \{1, \ldots, N\}\}, \hspace{1cm} 
\hat q^{\max} \coloneqq \max\{\hat q_{i,1}:i \in \{1, \ldots, N\}\}.
\end{align*}
$"\Rightarrow:"$
Since
\begin{align*}
    S_2^{-1} = \frac{1}{1-a-b}\begin{pmatrix}
        1-b & -a \\
        -b & 1-a
    \end{pmatrix},
\end{align*}
it has to hold
\begin{align*}
    \frac{p_{1,m}(1-b) - a p_{2,m}}{1-a-b} \in [0,1], 
     \frac{p_{1,m}(-b) + (1-a) p_{2,m}}{1-a-b} \in [0,1].
\end{align*}
For $\hat p_{1,m_1} = 1, \hat p_{2,m_1} < 1$ it has to hold 
\begin{align*}
0 \leq \frac{1 - b - a \hat p_{2,m_1}}{1-a-b} \leq 1
\end{align*}
which consequences $a \leq 0$. For $\hat p_{1,m_2} < 1, \hat p_{2,m_2} = 1$ it has to hold 
\begin{align*}
0 \leq \frac{1-a - b \hat p_{1,m_2}}{1-a-b} \leq 1. 
\end{align*}
This yields to $b \leq 0.$  Furthermore, since $\hat q^{\max} = 1, \hat q^{\min} = 0$, it holds $0 \leq a \leq 1, 0 \leq b \leq 1.$ Hence, $a = 0, b = 0$. Analogously, we can prove the same claim for all other possible combinations of allele frequencies that fulfill the constraints. \\
$"\Leftarrow"$ We need the following notation:
\begin{align*}
\hat u_\ast & := \min\left\{ \frac{\hat p_{2,j,m}}{\hat p_{1,j,m}}: m=1,...,M, j=1,2\right\}, \\
\hat u^\ast & := \max\left\{ \frac{\hat p_{2,j,m}}{\hat p_{1,j,m}}: m=1,...,M, j=1,2\right\}, \\
\hat v_\ast & := \min\left\{ \frac{\hat q_{i,1}}{\hat q_{i,2}}: i=1,...,N\right\},\hspace{1cm} \hat v^\ast := \max\left\{ \frac{\hat q_{i,1}}{\hat q_{i,2}}: i=1,...,N\right\}.
\end{align*} We know from Theorem 1 in \cite{heinzel2025revealing}
\begin{itemize}
    \item[a)] $\frac{u^\ast - 1}{u^\ast + v_\ast} = 1$ \hspace{1cm} b) $\frac{1 - u^\ast}{1 + u^\ast / v_\ast} = 0$
\hspace{1cm} $\frac{1 + v^\ast}{u_\ast + v^\ast} = 1$ \hspace{1cm} d)$\frac{1 + v^\ast}{1 + v^\ast / u_\ast} = 0.$
\end{itemize}

We know from $b)$ that $\hat q^{\min} = 0$, from $a)$ we infer that there exists $j, m_1$ with $p_{2,j,m_1} = 0$, from $d)$ we conclude that there exists $j, m_1$ with $p_{1,j,m_1} = 0$ and we also know from $c)$ that $\hat q^{\max} = 1.$ Hence, the claim is proven.

\end{proof}

We use Corollary \ref{lemma:unique} to name constraints that consequence that the only possible matrices $S_K$ are in $\mathcal P(E_K)$ for $K \geq 3.$

\begin{corollary}[Possible Matrices $S_K$ for arbitrary $K$]\label{cor:unique}
    Let $K \geq 2$ be arbitrary. If for every $k \in \{1, \ldots, K\}$ there is at least one individual with
   $\hat q_{i,k} = e_k$ 
and for every $k,j \in \{1, \ldots, K\}, j \neq k,$ there exists at least one marker with $\hat p_{k,m} = e_k + a_{j, m} e_{j}$  with $0 < a_{j,m} < 1$ and $a_{i,m_1} \neq a_{j,m_2} (i \neq j \lor m_1 \neq m_2)$, the only possible matrices are $S_K \in \mathcal P(E_K)$.
\end{corollary}
\begin{proof}
We prove the claim through induction. The constraints yield to the equation system
 \begin{align*}
    \hat q_1 \hat p_1 &= 1, \hat q_1 \hat p_2 = 0,\ldots, \hat q_1 \hat p_K = a_1 \\
     \hat q_2 \hat p_1 &= a_{2}, \hat q_2 \hat p_2 = 1,\ldots, \hat q_2 \hat p_K = 0 \\
         \hat q_3 \hat p_1 &= 0, \hat q_3 \hat p_2 = a_3,\ldots, \hat q_3 \hat p_K = 0 \\
     &\vdots \ldots \\
     \hat q_K \hat p_1 &= 0,\hat q_K \hat p_2  = 0, \ldots ,\hat q_K \hat p_K = 1.
\end{align*}
\underline{ Induction Beginning:} $K = 2:$\\
This can be inferred from Corollary \ref{lemma:unique}.

\noindent
\underline{Induction Step:} $K + 1$ Populations \\
We know, because of the at least $K - 2$ zeros per $\hat q_i$, that maximal two entries of the allele frequencies can be non-zero. However, they cannot be equal, because of the induction assumption. Hence, one value has to equal to one and the other one smaller than one. This implies the claim.
\end{proof}

Similarly to \cite{pfaffelhuber2022central}, we can name a constraints that holds if and only if the MLE in the supervised setting is unique.

\begin{lemma}[Uniqueness in the supervised setting]
    Constraint (*) in Assumption \ref{ass:invert} holds if and only if, the MLE $\hat Q^N$ is unique in the supervised setting.  
\end{lemma}

We only considered the uniqueness of the MLE for finite $M,N.$ However, for $M,N \to \infty,$ we only see differences between the MLEs, if we change at least $\lceil \beta M \rceil$ markers and $\lceil \alpha N \rceil$ individuals with $\alpha, \beta > 0.$ Hence, the above theory only states that if the constraints in Corollary \ref{cor:unique} hold, a change in $\alpha M$ markers and in $\beta N$ individuals implies a change of the log-likelihood.

\subsection{Consistency of the MLE}

In this section, we prove the consistency of the MLEs in the unsupervised setting. Therefore, we adapt the ideas of Hoadley's proof \citep{hoadley1971} to make it applicable to the MLEs in the Admixture Model. Therefore, we need  Notation \ref{not}.

\begin{notation}\label{not}
We define 
    \begin{align*}\
X^{(B)} &:= X \vee (-B), \\
R_{i,m}(\theta) &= 
        \begin{cases}
            \ell((q,p)|X_{i,m}) - \ell((q^0, p^0)|X_{i,m}), \textup{ if } \textup{exp}(\ell(\theta|Y_{\ell})) > 0 \\
            0, \textup{ otherwise,}
        \end{cases} \\
R_{i,m}((q,p), \rho) &= \sup\{R_{i,m}(t): d(t, (q,p)) \leq \rho\}, \\
\bar r_{N, M}((q,p)) &= \frac 1 {NM} \sum_{i = 1}^N \sum_{m=1}^M \mathbb E(R_{i,m}((q,p)), \\
\Theta_{d}(\epsilon) &:= \left\{(q,p) \in \Theta: \min\{d\left(\left(q,p\right), \mathcal M\right)\} \geq \epsilon\right\}.
\end{align*}
\end{notation}

We use the following conditions to prove consistency, which are similar to the ones by \cite{hoadley1971}.

\begin{definition}[Conditions for consistency]\label{th:hoadley}
We define the conditions 
        \begin{itemize}
        \item[C2] $\exp(\ell((q,p)|X_{i,m}))$ is almost sure an upper semi continuous function of $(q,p),$ uniformly in $(i,m)$,
        \item[C3 (i)]{ There exists $\rho^* = \rho^*(\theta)> 0, \delta > 0, r > 0, \gamma > 0$ with 
 $\forall 0 \leq \rho \leq \rho^*$ it holds $\mathbb E\left(R^{(0)}_{i,m}((q,p), \rho)^{1+\delta}\right) \leq \gamma$,}
        \item[C4 (i)] For all $(q,p) \in \Theta_d(\epsilon)$, it holds $\lim_{M,N \to \infty} \bar r_{N,M}((q,p)) < 0$,
        \item[C5] $R_{i,m}((q,p), \rho), V_{i,m}(r)$ are measurable functions of $X_{i,m}.$
    \end{itemize}

\end{definition}
Our parameter space is bounded anyway, which simplifies the proof compared to the proof in \cite{hoadley1971}. Another difference is that \cite{hoadley1971} assumed that the parameter space is finite dimensional while we do not assume it. Hence, we have to choose the metric carefully.

Furthermore, we need, as \citet{hoadley1971}, the following two lemmas to prove consistency.
\begin{lemma}\label{lemma:RM}
The convergence
        \begin{align}
       \lim_{M,N \to \infty} \mathbb P\left(\sup\left\{\sum_{m = 1}^M \sum_{i = 1}^N R_{i,m}((q,p)):(q,p) \in \Theta_{d}(\epsilon)\right\} \geq c_p\right) = 0 \label{con}
    \end{align}
   for any $c_p \in \mathbb R$ implies
    \begin{align*}
       \lim_{M, N \to \infty} \mathbb P\left(\left(\hat Q^N, \hat P^M\right) \in \Theta_{d}(\epsilon)\right) = 0.
    \end{align*}
\end{lemma}

\begin{lemma}\label{lemma:uniformly}
    Conditions C2, C3(i) and C4(i) imply that there exists for each $(q,p) \in \Theta_d(\epsilon)$ a $\rho((q,p)) \leq \rho^*$ for which
    \begin{align*}
        r_{i,m}^{(B)}((q,p), \rho(q,p)) < r_{i,m}^{(B)}((q,p)) - \bar r^{(B)}((q,p))/2. 
\end{align*}

\end{lemma}

Additionally, we need a lemma about the compactness of the set $\Theta_{d}(\epsilon).$

\begin{lemma}\label{lemma:compact}
  The set $\Theta_{d}(\epsilon)$ is compact with respect to the metric $d.$
\end{lemma}
\begin{proof}
   The set $\Theta$ is compact with respect to the metric $d$ according to the Theorem of Tychonoff. Moreover, $\Theta_{d}(\epsilon)$ is a closed subset of the compact set, which implies the claim. The latter holds since 
   $$\Theta_d(\epsilon) = \{(q,p) \in \Theta: \min\{d((q,p), \mathcal M)\} \in [\epsilon, 2]\},$$
   i.e. the set is an inverse image of a closed set for a continuous function.
\end{proof}

Finally, we prove weak consistency of the MLEs in the Admixture Model.

\begin{proof}[Proof of Theorem \ref{th:wc_us}]
According to Lemma \ref{lemma:RM}, we just have to prove that equation \eqref{con} holds. Therefore, we check whether constraints $C2, C3(i), C4(i)$ and $C5$ hold.  Conditions $C2$ and $C5$ are trivial.

\begin{itemize}
    \item[C3) (i)] We choose $\delta = 1$ and define $D\left(\rho, q_{i, \cdot}^0, p_{\cdot, m}^0\right) := \sup\{X_{i,m} \ln(\langle t_1, t_2\rangle) + (2-X_{i,m}) \ln(1-\langle t_1, t_2\rangle): d\left((t_1, t_2), \left(q^0, p^0\right)\right) \leq \rho, (t_1, t_2) \in \Theta_{d}(\epsilon)\}.$ Note that it holds $D\left(\rho, q_{i, \cdot}^0, p_{\cdot, m}^0\right) \leq 0$. We have
    \begin{align*}
        &\left(R_{i,m}^{(0)}((q,p), \rho)\right)^2 \\
        &= \left(\left(D\left(\rho, q_{i, \cdot}^0, p_{\cdot, m}^0\right) - X_{i,m} \ln\left(c^0_{i,m}\right) - (2-X_{i,m}) \ln\left(1-c^0_{i,m}\right)\right)^{(0)}\right)^2  \\
        &\begin{cases}
           = 0, \text{ if } D(\rho, q^0_{i, \cdot}, p^0_{\cdot, m}) < X_{i,m} \ln(c^0_{i,m}) + (2-X_{i,m}) \ln(1-c^0_{i,m})\\
            \leq (- X_{i,m} \ln(c^0_{i,m}) - (2-X_{i,m}) \ln(1-c^0_{i,m}))^2 , \text{ else}
        \end{cases} \\
        &\leq (X_{i,m} \ln(c^0_{i,m}) + (2-X_{i,m}) \ln(1-c^0_{i,m}))^2.
    \end{align*}
    Hence, 
    \begin{align*}
        &\mathbb E\left(R^{(0)}_{i,m}(\theta, \rho)^2\right) \\
        &\leq  \mathbb E\left((X_{i,m} \ln(c^0_{i,m}) + (2-X_{i,m}) \ln(1-c^0_{i,m})))^2\right) \\
        &=  \mathbb E\left((X^2_{i,m} \ln^2(c^0_{i,m}) + (2-X_{i,m})^2 \ln^2(1-c^0_{i,m}) + 2X_{i,m} \ln(c_{i,m}^0) \ln(1-c_{i,m}^0)\right) \\
        &= \left(c^0_m\right)^2 4\ln(c^0_{i,m})^2 + 2c^0_{i,m} \left(1-c^0_{i,m} \right)\ln(c^0_{i,m})^2 \ln(1 - c^0_{i,m})^2 \\
        & + 4 \left(1-c^0_{i,m} \right)^2 \ln(1 - c^0_{i,m})^2\\
        & \leq 4 + 1+4 \leq \gamma  < \infty.
    \end{align*}
    \item[C4) (i)] It holds \begin{align*}
    \frac 1 2\mathbb E(\ell(q,p|X_{i,m}) - \ell(q^0, p^0|X_{i,m})) 
    &= c^0_{i,m} \ln\left(\frac{c_{i,m}}{c^0_{i, m}}\right) +  (1-c^0_{i,m})\ln\left(\frac{1-c_{i,m}}{1-c^0_{i,m}}\right).
\end{align*}
With the derivation with respect to $c_{i,m}$
\begin{align*}
 \frac{\partial}{\partial c_{i,m}} \frac 1 2\mathbb E(\ell(q,p|X_{i,m}) - \ell(q^0, p^0|X_{i,m})) = \frac{c^0_{i,m}}{c_{i,m}} - \frac{1-c^0_{i,m}}{1-c_{i,m}},
\end{align*}
we see that $c_{i,m} = c^0_{i,m}$ is the only maximal point of
$c_{i,m} \mapsto \mathbb E(\ell(q,p|X_{i,m}) - \ell(q^0, p^0|X_{i,m})).$
Hence, 
$$ \frac 1 2\mathbb E(\ell(q,p|X_{i,m}) - \ell(q^0, p^0|X_{i,m})) \begin{cases}
        < 0, c_{i,m}\neq c^0_{i,m} \\
        = 0, c_{i,m}= c^0_{i,m}
    \end{cases}$$
and because of our definition of the set $\mathcal M$, for any $\epsilon > 0$, there exists a constant $c( \epsilon) < 0$ with 
\begin{align*}
 \frac{1}{M} \frac{1}{N} \sum_{m = 1}^M \sum_{i = 1}^{N}\mathbb E(\ell(q,p|X_{i,m}) - \ell(q^0, p^0|X_{i,m})) 
\xrightarrow[]{M, N \to \infty} c(\epsilon) < 0   
\end{align*}
for all $(q, p) \in \Theta_{d}(\epsilon)$.
\end{itemize}

From C2, C3 (i) and C4(i) we infer, according to Lemma \ref{lemma:uniformly}
\begin{align}
        r_{i,m}^{(B)}((q, p), \rho(q, p)) < r_{i,m}^{(B)}((q,p)) - \bar r^{(B)}((q,p))/2 \label{3.5}
\end{align}
for all $(q,p) \in \Theta_{d }(\epsilon).$ Let $S(c, r)$ be an open ball with middle point $c$ and radius $r.$
Because of Lemma \ref{lemma:compact}, there exist $q^1, \ldots, q^g, p^1, \ldots, p^g \in \Theta_{d}(\epsilon)$ for which 
\begin{align*}
    \Theta_{d}(\epsilon) \subset \bigcup_{j = 1}^g S\left((q^j, p^j), \rho\left((q^j, p^j)\right)\right). 
\end{align*}
From \eqref{3.5} and C4(i) we know
\begin{align*}
    \textup{lim sup}_{N, M \to \infty} \bar{r}_{N,M}^{(0)}\left((q^j, p^j), \rho\left((q^j, p^j)\right)\right) < 0. 
\end{align*}
Hence and because of C3(i), we can apply Theorem A4 in \cite{hoadley1971} to $R_{i,m}((q^j, p^j), \rho\left((q^j, p^j)\right)), j = 1,..., g.$
Consequently, it holds
\begin{align*}
       \lim_{M,N \to \infty} \mathbb P\left( \sum_{i = 1}^N \sum_{m = 1}^MR_{i,m}\left((q^j, p^j), \rho\left((q^j, p^j)\right)\right) \geq c_p\right) = 0
\end{align*}
for any $c_p \in \mathbb R.$ Since
\begin{align*}
       &\sup\left\{\sum_{i=1}^N \sum_{m=1}^M R_{i,m}((q,p)): (q,p) \in \Theta_{d}(\epsilon)\right\} \\
       &\hspace{1cm} \leq \max_{j \in \{1, \ldots, g\}}\left\{\sum_{i = 1}^N \sum_{m = 1}^MR_{i,m}\left((q^j, p^j), \rho\left((q^j, p^j)\right)\right)\right\},
\end{align*}
the claim is proven.

\end{proof}

Moreover, the consistency with respect to the euclidean norm for finite dimensions is a direct consequence of the Theorem that is stated in \cite{hoadley1971}.

\begin{corollary}
We only consider one single individual. Let the allele frequencies be known. We assume in addition that 
$$\textup{(AU)}\, \lim_{M \to \infty} \frac{1}{M} \sum_{m=1}^M c^0_{i,m} \ln\left(\frac{c_{i,m}}{c^0_{i, m}}\right) +  (1-c^0_{i,m})\ln\left(\frac{1-c_{i,m}}{1-c^0_{i,m}}\right) < 0$$
for every $q \neq q^0.$
Then, it holds
\begin{align*}
    \lim_{M,N \to \infty} \mathbb P\left(|\hat Q^{M} - q^0| \geq \epsilon\right) =0
\end{align*}
for any $\epsilon > 0.$
\end{corollary}

Constraint (AU) seems at the first glance difficult to prove. However, for example $K = 2$ every $p_{\cdot,m}$ with $p_{\cdot, m}  = p_{\cdot, m+1}$ and $p_{1,m} \neq p_{2,m}$ fulfills the constraint. Then, the condition holds. More generally, we can apply  Pinsker's inequality (which is widely considered, see e.g. \cite{fedotov2003refinements}), i.e. it holds
    \begin{align*}
        c_m^0 \ln(c_m/c_m^0) + (1-c^0_m) \ln((1-c_m)/(1-c_m^0)) 
        &\leq - \frac{1}{2} |c_m - c_m^0|_2^2.
    \end{align*}
    Hence for $K = 2$, if 
    \begin{align*}
       \frac{1}{M} \sum_{m=1}^M  (p_{2,m} - p_{1,m})^2 > 0,
    \end{align*}
    the condition is fulfilled. A trivial example for this is if $|p_{2,m} - p_{1,m}| \geq \delta$ for a $\delta > 0.$
Moreover, constraint (AU) covers cases in which the MLE is non-unique in the supervised setting. An example for this, as mentioned in (\cite{pfaffelhuber2022central}), would be $p_{\cdot, m} = (p_{1,m}, p_{2,m}, \alpha p_{1,m} + (1-\alpha) p_{2,m})$ for all $m \in \{1,..., M\}.$ For these cases, both $q = (0,0,1)$ and $q = (\alpha, 1-\alpha, 0)$ yield to the same likelihood.

\subsection{Central Limit Results}
In this chapter, we prove central limit theorems using the consistency of the MLEs, either when the parameter space is open or when the true parameter lies on the boundary of the parameter space. Therefore, we write $A \succ 0$ for a positive definite matrix $A.$ Please note that we will consider a $K-1$ dimensional parameter space instead of a $K$ dimensional parameter space. However, the claims for both parameter spaces are equivalent.

\subsubsection{Open Parameter space}
\citet{hoadley1971} (Theorem 2) has already proven a theorem that names conditions for the asymptotic normality of the MLE for independently, non-identically distributed random variables. This theorem is stated below.

\begin{theorem}[Central Limit Result for the MLE]\label{app:CLT}
We name the derivation of the log likelihood with respect to the $k^{th}$ and the $h^{th}, j^{th}$ parameter by $\dot \Phi_k(X_{i,m}, (q, p)), \ddot{\Phi}_{hj}(X_{i,m}, (q, p))$ , respectively.
The conditions
    \begin{itemize}
\item[N1] $\Theta$ is an open subset of $\mathbb R^{(N_C \times (K-1)) \times (K \times M_C)}$,
\item[N2] It holds $\left(\hat q^{N_C}, \hat p^{M_C}\right) \xrightarrow[]{M, N\to \infty}_p \left(q^0, p^0\right),$
\item[N3] $\dot{\Phi}(X_{i,m}, (q,p)), \ddot{\Phi}(X_{i,m}, (q,p))$ exist almost surely,
\item[N4] $\ddot{\Phi}(X_{i,m}, (q,p))$ is a continuous function of $(q,p)$, uniformly in $(i,m)$ and a measurable function of $X_{i,m},$
\item[N5] It holds $\mathbb E_\theta(\dot{\Phi}_\ell(X_{i,m}, (q,p))) = 0 \forall \ell$,
\item[N6] $\Gamma_\ell(\theta) = \mathbb E_\theta(\dot{\Phi}(X_{i,m}, (q,p)) \dot{\Phi}(X_{i,m}, (q,p))^\intercal) = - \mathbb E_\theta(\ddot{\Phi}(X_{i,m}, (q,p))),$
\item[N7] It holds $\frac{1}{MN}\sum_{i=1}^N\sum_{m=1}^M   \Gamma_{i,m}(\tilde q, \tilde p)\xrightarrow[]{D \to \infty} \bar \Gamma(\tilde q, \tilde p)$, $\bar \Gamma(\tilde q, \tilde p) \succ 0 $,
\item[N8] There exists $\delta > 0$ with 
$$ \sum_{i=1}^N \sum_{m=1}^M \mathbb E(|\lambda^\intercal\dot{\Phi}(X_{i,m}, (q^0,p^0))|^{2+\delta})/(MN)^{(2+\delta)/2} \xrightarrow[]{D \to \infty} 0$$ for all $\lambda \in \mathbb R^{(N_C \times (K-1)) \times (K \times M_C)}$.
\item[N9] There exists $\gamma, \epsilon > 0$ and random variables $B_{\ell,hj}(Y_k)$ with \\
(i) $\sup\{|\ddot{\Phi}_{hj}(X_{i,m}, t)|: \|t - (q^0, p^0)\| \leq \epsilon\} \leq B_{\ell,hj}(X_{i,m})$,\\
(ii) $\mathbb E(|B_{\ell,hj}(X_{i,m})|^{1+\delta}) \leq \gamma$,
    \end{itemize}
imply
\begin{align*}
   \left( \sqrt{M}(\hat Q^{N_C} - q^0), \sqrt{N}(\hat P^{M_C} - p^0)\right) \xRightarrow{M, N \to \infty} \mathcal N\left(0, \Gamma^{-1}\left(\tilde q^0, \tilde p^0\right)\right).
\end{align*}
\end{theorem}

Now, we are  ready to prove a CLT for the MLEs.

\begin{proof}[Proof of Theorem \ref{th:clt_unsupervised}]
We apply Theorem \ref{app:CLT}. Conditions $(N1) - (N4), (N7)$ are trivial. Conditions (N5) and (N6) hold, since we can change derivation and integral of $\ell(X_{i,m}, (q,p)).$
    \begin{itemize}
\item[(N8)] We only consider the derivation of $\ell(\tilde q, \tilde p|X)$ with respect to $\tilde p$. The claim for the derivation of the log-likelihood with respect to $\tilde q$ can be inferred in the same way. For $c^0_{i,m} \in (0,1), m \in \{1, ..., M_C\}$, it holds 
\begin{align*}
    \mathbb E\left(|\lambda ^{\top} \frac{d}{d\tilde p_{\cdot, m}} \ell(\tilde q^0,\tilde p^0|X_{i,m})|^3\right) 
    &= \mathbb E\left(\left|\lambda^{\top} \left( \frac{X_{i,m} \tilde q^0_{i,\cdot}}{\langle \tilde q^0_{\cdot}, \tilde p^0_{\cdot m} \rangle} -  \frac{(2- X_{i,m}) \tilde q^0_{i,\cdot}}{1 - \langle \tilde q^0_{\cdot}, \tilde p^0_{\cdot m} \rangle}\right)\right|^3\right) \\
&\leq c(\lambda) \mathbb E\left(\left| \frac{X_{i,m} - 2\langle \tilde q^0_{\cdot}, \tilde p^0_{\cdot m} \rangle}{\langle \tilde q^0_{\cdot}, \tilde p^0_{\cdot m} \rangle (1 - \langle \tilde q^0_{\cdot}, \tilde p^0_{\cdot m} \rangle}\right|^3\right) \\
&\leq c(\lambda) \mathbb E\left(\left(\frac{(2, ..., 2)}{\langle \tilde q^0_{\cdot}, \tilde p^0_{\cdot m} \rangle (1 - \langle \tilde q^0_{\cdot},\tilde p^0_{\cdot m} \rangle}\right)^3\right).
\end{align*}
Hence, we have
\begin{align*}
   \sum_{m = 1}^M \frac{\mathbb E(|\lambda^{\top} \dot{\Phi}(X_m,(q,p))|^3)}{M^{3/2}}
&\leq  c(\lambda) \sum_{m = 1}^M\frac{(8, ..., 8)}{(\langle q^0_{\cdot}, p^0_{\cdot m} \rangle)^3 (1 - \langle q^0_{\cdot}, p^0_{\cdot m} \rangle)^3M^{3/2}} \\
&
\xrightarrow[]{M \to \infty} 0,
\end{align*}
because of (***) in Assumption \ref{ass:invert}.
\item[(N9)] It holds
\begin{align*}
   - \mathbb E\left(\frac{\partial^2}{\left(\partial q^0_{i, \cdot}\right)^2} \ell(q^0, p^0|X_{i,m})\right) &= \begin{cases}\left(\frac{ 2  \underline{p^0_{s,m}}}{\langle q^0_{\cdot}, p^0_{\cdot m}\rangle}  + \frac{ 2  \underline{p^0_{s,m}}}{1 - \langle q^0_{\cdot}, p^0_{\cdot m}\rangle}\right), \textup{ if } \langle q^0_{i, \cdot}, p^0_{\cdot, m} \rangle \in (0,1) \\
  2  \underline{p^0_{s,m}} , \textup{ else}.  
\end{cases}
\end{align*}
Then, the claim is a direct consequence of our choice of the parameter space.
\end{itemize}
\end{proof}

So far, we just assumed that the Fisher Information is invertible. Let us consider this constraint in more detail:
    The matrix $\bar \Gamma\left(\tilde q^0, \tilde p^0\right)$ is a block matrix, i.e. we just have to ensure $\bar \Gamma_{var(q)}(\tilde q, \tilde p) \succ 0, \bar \Gamma_{var(p)}(\tilde q, \tilde p) \succ 0$. Therefore, we exclusively focus on $\bar \Gamma_{var(q)}(\tilde q, \tilde p)$. 
    First, we prove that $\bar \Gamma_{var(q)}(\tilde q, \tilde p) \succ 0$ exists. Recall the definition of $\underline{p_{K m}}$ and
$\underline{p_{s,m}}$ from Assumption \ref{ass:invert}.

Constraint (***) ensures that the limit exists. Hence, it remains to prove $\bar \Gamma_{var(q)}(\tilde q, \tilde p) \succ 0$. 
It holds
\begin{align*}
    \sum_{m: p_{\cdot, m} \in S_j^p} - \mathbb E\left(\frac{\partial^2}{\partial^2 q^0_{i, \cdot}} \ell(q^0, p^0|X_{i,m})\right) \succ 0 \, \forall j = 1,2,....
\end{align*}
Unfortunately, we cannot conclude from this to the positive-definiteness of the whole matrix. Note, however, for the application, this does not matter since the sum $ - \frac{1}{M}\sum_{m = 1}^M  \mathbb E\left(\frac{\partial^2}{\partial^2 q^0_{i, \cdot}} \ell(q^0, p^0|X_{i,m})\right)$
is for $M \in \mathbb N$ always positive definite under our constraints. Furthermore, if we restrict our parameter space to $(q,p): q_{i, k} > \epsilon, p_{k,m} > \epsilon,$ we can infer the positive definiteness of the matrix.

\subsubsection{Closed Parameter Space}

Consideration of a closed parameter space is in particular important, as the uniqueness of the MLE requires that there are some true ancestries and true allele frequencies on the boundary of the parameter space. First, we name the theorem that holds according to \cite{andrews1999estimation} for MLEs, if they are on the boundary of the parameter space. Afterwards, we prove that the constraints of this theorem hold for the MLE in the Admixture Model. 

\begin{definition}[Cone]
    A set $\Lambda \subset \mathbb R^\rho$ is a cone if 
    $$\left(\lambda \in \Lambda\right)\Rightarrow \left(a \lambda \in \Lambda \, \forall a \in \mathbb R_+\right).$$
\end{definition}

\begin{definition}[Locally Equal]
    A set $\Lambda$ is locally equal to a set $\mathcal S$ if it holds $\mathcal S \cap C(0, \epsilon) = \Lambda \cap C(0, \epsilon)$ for some $\epsilon > 0$ with an open cube $C(0,\epsilon).$
\end{definition}

To understand the theorem for a closed parameter space, we define 
\begin{align*}
    \ell(q,p|X) &= \ell(q,p|X) + \left((q, p)-\left(q^0, p^0\right)\right) \frac{\partial}{\partial (q, p)} \ell(q,p|X) \\
    & \hspace{1cm}- \frac 1 2  \left((q, p)-\left(q^0, p^0\right)\right)^2 \frac{\partial^2}{\partial (q, p)^2}\ell(q,p|X) + R_{M}(q^0|X).
\end{align*}

With this notation, we can name the theorem. 

\begin{theorem}[CLT if the parameter is on the boundary of the parameter space]\label{th:andrews}
Let $\dot{\Phi}(X_m, q^0), \ddot{\Phi}(X_m, q^0)$ be the first and second derivations of the log-likelihood function with respect to $q$, respectively. We recall the definition of $Z, \hat \lambda$ in Theorem \ref{th:uv_closed}.
The constraints
\begin{itemize}
  \item[A1] $\lim_{M \to \infty} \sup_{\theta \in \Theta: \|q-q^0\| \leq \gamma_M} \frac{|R_M(\theta|X)|}{(1 + \|\sqrt{M}(q-q^0)\|)^2} = 0$ for all $\gamma_M \to 0$,
    \item[A2] $\left(\frac{1}{\sqrt{M}} \dot{\Phi}(X_1,..., X_M, q^0), \frac 1 M \ddot{\Phi}(X_1,..., X_M, q^0)\right) \xRightarrow[]{M \to \infty} (Z, \bar{\Gamma}(q^0))$ with a non-singular and symmetric matrix $F,$
    \item[A3] $\sqrt{M}\left(\hat Q^M - q^0\right) = O_p(1)$,
    \item[A4] The set $\Theta - q^0$ is locally equal to a cone $\Lambda \subset \mathbb R^K.$
    \item[A5] $\Lambda$ is convex.
\end{itemize}
imply $\sqrt{M}(\hat Q^M - q^0) \xRightarrow[]{M \to \infty} \hat \lambda.$
\end{theorem}

Altogether, we have defined everything to prove a CLT if the true parameter is on the boundary of the parameter space, i.e. to prove Theorem \ref{th:uv_closed}.

\begin{proof}[Proof of Theorem \ref{th:uv_closed}]
We prove that the constraints in Theorem \ref{th:andrews} hold. Therefore, without loss of generality, we assume $q^0_1 = 1.$
\begin{itemize}
    \item[A1] It holds with probability one
   \begin{align*}   
   |R_M(q|X_m)| &\leq \begin{cases}
       \left(q- q^0\right)^3 \left(\frac{2}{\langle q, p^0_{\cdot, m} \rangle^3} +\frac{2}{(1-\langle q, p^0_{\cdot, m} \rangle)^3}\right) , \langle q, p^0_{\cdot, m} \rangle \in (0,1) \\
       \left(q - q^0\right)^3  \frac{2}{\langle q, p^0_{\cdot, m} \rangle^3},   \langle q, p^0_{\cdot, m} \rangle = 1 \\
              \left(q - q^0\right)^3 \frac{2}{(1-\langle q, p^0_{\cdot, m} \rangle)^3},  \langle q, p^0_{\cdot, m} \rangle = 0.
   \end{cases}
    \end{align*}    
    Here, $(q - q^0)^3 := \sum_{k=1}^K \left(q_k-q_k^0\right)^3.$
    Hence, it holds 
    \begin{align*}
        &\lim_{M \to \infty} \sup_{q \in \Theta: |q-q^0| \leq \gamma_M} \frac{1}{M} \sum_{m=1}^M |R_M(q|X_m)| \\
        &\leq \lim_{M \to \infty} \sup_{q \in \Theta: |q-q^0| \leq \gamma_M} \frac{\left(q - q^0\right)^3}{M}\sum_{m=1}^M \frac{2}{\langle q, p^0_{\cdot, m} \rangle^3}1_{\langle q, p^0_{\cdot, m} \rangle^3 > 0} + \frac{2}{1 - \langle q, p^0_{\cdot, m} \rangle^3} 1_{\langle q, p^0_{\cdot, m} \rangle^3 <1}\\
        &\leq 2\lim_{M \to \infty} \frac{\gamma_M^3}{M} \sum_{m=1}^M c(\epsilon) \\
        &= 2 c(\epsilon) \lim_{M \to \infty} \gamma_M = 0.
    \end{align*}
 Here, we used the assumption that $\langle q^0, p_{\cdot,m}^0\rangle \in [\epsilon, 1-\epsilon]$ for all except finitely many markers.
\item[A2] 
    The claim is implied by the CLT and the law of large numbers. Specifically, we choose $B_M := \sqrt{M}$. Then, it holds (according to Theorem A6 in \cite{hoadley1971}, which holds since N5, N6, N7 and N8 are fulfilled)
    \begin{align*}
         \frac{1}{\sqrt{M}} \sum_{m = 1}^M  \dot{\Phi}(X_m,q^0) \xRightarrow[]{M \to \infty} \mathcal N(0, \bar{\Gamma}(q^0)).
    \end{align*}
 Moreover, it holds
       \begin{align*}
         \frac{1}{M} \sum_{m = 1}^M  \ddot{\Phi}(X_m,q^0) \xrightarrow[]{M \to \infty}_p \bar{\Gamma}(q^0).
    \end{align*} 
    The limit exists and is positive definite because of our assumption.
    The second claim can be proven by
    \begin{align*}
        \mathbb P \left(|\frac{1}{M} \sum_{m = 1}^M  \ddot{\Phi}(X_m,q^0) - \bar \Gamma(q^0)| \geq \epsilon \right) &\leq \frac{\mathbb V\left( \sum_{m = 1}^M  \ddot{\Phi}(X_m,q^0)\right)}{M^2 \epsilon} \\
        &= \frac{\sum_{m = 1}^M   \mathbb V\left(\ddot{\Phi}(X_m,q^0)\right)}{M^2 \epsilon} \\
        &\leq  \frac{\sum_{m = 1}^M   \mathbb E\left(\left(\ddot{\Phi}(X_m,q^0) \right)^2\right)}{M^2 \epsilon} \\
        &\leq 2 \frac{\sum_{m = 1}^M   \frac{1}{c_m^0} + \frac{1}{1-c_m^0}}{M^2 \epsilon}\\
        &\xrightarrow[]{M \to \infty} 0,
    \end{align*}
    because of the assumption that $c^0_m \in [0, \epsilon] \cap [1-\epsilon, 1]$ only holds for finitely many $m$ and because of $\mathbb E(X_{i,m}^2) = 2 c_{i,m}^0 \left(1-c_{i,m}^0\right).$
    \item[A3] The weak consistency of the MLE, (A1) and (A2) imply (A3) according to \cite{andrews1999estimation}.
    \item[A4] Obviously, the set $\Lambda$ is a cone. Hence, we just have to prove that there exists an open cube $C(0, \epsilon)$ with
    \begin{align*}
        \left(\Theta - q^0\right) \cap C(0, \epsilon) =  \Lambda \cap C(0, \epsilon).
    \end{align*}
    Here, we choose $\epsilon = 2$ to fulfill this constraint. Then, we have
        \begin{align*}
        &\left(\Theta - q^0\right) \cap C(0, 2) \\
        &= \left\{v \in \mathbb R^{K-1}: v_1 \in[-1,0], v_i \in [0, 1], |v_1+\ldots + v_{K-1}| \leq 1, i = 2, \ldots, K-1\right\} \\
        & \hspace{1cm} \cap (-1, 1) \times \ldots \times (-1, 1) \\
        &= \left\{v \in \mathbb R^{K-1}: v_1 \in(-1,0], v_i \in [0, 1), |v_1+\ldots + v_{K-1}| \leq 1, i = 2, \ldots, K-1\right\}\\
        &= \Lambda \cap C(0, 2).
    \end{align*}
    \item[A5] Let $\lambda_i \in \Lambda$ for all $i \in \{1, \ldots, K-1\}$ and $\alpha \in [0,1].$ Then, it holds $\sum_{i = 1}^{K-1} \alpha \lambda_i + (1-\alpha) \lambda_i = \alpha \sum_{i = 1}^{K-1}  \lambda_i + (1-\alpha) \sum_{i = 1}^K \lambda_i =0$ and 
    $\alpha \lambda_i + (1-\alpha) \lambda_i 
       \geq 0$ for $ i \geq 2$ and $\alpha \lambda_1 + (1-\alpha) \lambda_1 \leq 0.$
\end{itemize}
\end{proof}

\section{Discussion}

In this study, we established limit results for the maximum likelihood estimators (MLEs) in the Admixture Model. Specifically, we analyzed the consistency of the MLEs in the unsupervised setting and proved central limit theorems (CLTs) for both cases—when the true parameter lies on the boundary of the parameter space and when it does not.

The CLTs presented in this work significantly extend previous approaches by \citet{pfaff2004information, pfaffelhuber2022central}, and they enable the quantification of uncertainty in MLEs under a wide range of conditions. To facilitate practical applications, we provide \texttt{Python} code on \href{https://github.com/CarolaHeinzel/Limit_Results_Admixture_Model}{GitHub}, covering all possible settings considered in our theory.

However, we emphasize that in the unsupervised setting, the MLEs are not unique. Therefore, a standard CLT for the MLE cannot be formulated in this case.

Our results can be applied in various ways. For example, they are useful in marker selection—an important topic in population genetics that has been widely studied \citep{phillips2014,kidd2014,xavier2022,xavier2020development, pfaffelhuber2020, resutik2023, phillips2019maplex, kosoy2009ancestry}, and are closely related to earlier work by \citet{pfaff2004information}. In addition, our theory can be used to assess the quality of a marker set in terms of estimation uncertainty.

We have also extended our framework to include the semi-supervised setting. This extension is particularly relevant for applications, as it allows for the estimation of uncertainty in the unsupervised case via known labels.

A natural next step is to consider scenarios where the random variables are not independently distributed. This situation frequently arises in practice, since genetic markers typically exhibit linkage disequilibrium. A suitable model for such dependencies is the Linkage Model, which incorporates a Hidden Markov Model to describe the correlation structure among markers, as introduced by \citet{falush2003}.

\bigskip 

\noindent {\bf {{Funding}}} \\
\\
\noindent
CSH is funded by the Deutsche Forschungsgemeinschaft (DFG, German Research Foundation) – Project-ID 499552394 – SFB Small Data. \\

\noindent {\bf {{Acknowledgment}}} \\
\\
CSH thanks Peter Pfaffelhuber for his continuous and strong support. \\

\noindent {\bf {{Declaration of Conflicts of Interest}}} \\
\\
The author declares that there are no conflicts of interest.

\begin{small}
    \bibliography{Bibliography}

\begin{thebibliography}{34}
\providecommand{\natexlab}[1]{#1}
\providecommand{\url}[1]{\texttt{#1}}
\expandafter\ifx\csname urlstyle\endcsname\relax
  \providecommand{\doi}[1]{doi: #1}\else
  \providecommand{\doi}{doi: \begingroup \urlstyle{rm}\Url}\fi

\bibitem[Alexander et~al.(2009)Alexander, Novembre, and Lange]{alexander2009fast}
David~H Alexander, John Novembre, and Kenneth Lange.
\newblock Fast model-based estimation of ancestry in unrelated individuals.
\newblock \emph{Genome research}, 19\penalty0 (9):\penalty0 1655--1664, 2009.

\bibitem[Andrews(1999)]{andrews1999estimation}
Donald~WK Andrews.
\newblock Estimation when a parameter is on a boundary.
\newblock \emph{Econometrica}, 67\penalty0 (6):\penalty0 1341--1383, 1999.

\bibitem[Badouin et~al.(2017)Badouin, Gouzy, Grassa, Murat, Staton, Cottret, Lelandais-Bri{\`e}re, Owens, Carr{\`e}re, Mayjonade, et~al.]{badouin2017}
H{\'e}l{\`e}ne Badouin, J{\'e}r{\^o}me Gouzy, Christopher~J Grassa, Florent Murat, S~Evan Staton, Ludovic Cottret, Christine Lelandais-Bri{\`e}re, Gregory~L Owens, S{\'e}bastien Carr{\`e}re, Baptiste Mayjonade, et~al.
\newblock The sunflower genome provides insights into oil metabolism, flowering and asterid evolution.
\newblock \emph{Nature}, 546\penalty0 (7656):\penalty0 148--152, 2017.

\bibitem[Behr et~al.(2016)Behr, Liu, Liu-Fang, Nakka, and Ramachandran]{behr2016}
Aaron~A Behr, Katherine~Z Liu, Gracie Liu-Fang, Priyanka Nakka, and Sohini Ramachandran.
\newblock Pong: fast analysis and visualization of latent clusters in population genetic data.
\newblock \emph{Bioinformatics}, 32\penalty0 (18):\penalty0 2817--2823, 2016.

\bibitem[Cabreros and Storey(2019)]{cabreros2019likelihood}
Irineo Cabreros and John~D Storey.
\newblock A likelihood-free estimator of population structure bridging admixture models and principal components analysis.
\newblock \emph{Genetics}, 212\penalty0 (4):\penalty0 1009--1029, 2019.

\bibitem[Divers et~al.(2011)Divers, Redden, Carroll, and Allison]{divers2011}
Jasmin Divers, David~T Redden, Raymond~J Carroll, and David~B Allison.
\newblock How to estimate the measurement error variance associated with ancestry proportion estimates.
\newblock \emph{Statistics and its interface}, 4\penalty0 (3):\penalty0 327, 2011.

\bibitem[Falush et~al.(2003)Falush, Stephens, and Pritchard]{falush2003}
Daniel Falush, Matthew Stephens, and Jonathan~K Pritchard.
\newblock Inference of population structure using multilocus genotype data: linked loci and correlated allele frequencies.
\newblock \emph{Genetics}, 164\penalty0 (4):\penalty0 1567--1587, 2003.

\bibitem[Fedotov et~al.(2003)Fedotov, Harremo{\"e}s, and Topsoe]{fedotov2003refinements}
Alexei~A Fedotov, Peter Harremo{\"e}s, and Flemming Topsoe.
\newblock Refinements of pinsker's inequality.
\newblock \emph{IEEE Transactions on Information Theory}, 49\penalty0 (6):\penalty0 1491--1498, 2003.

\bibitem[Ferguson(2017)]{ferguson2017}
Thomas~S Ferguson.
\newblock \emph{A course in large sample theory}.
\newblock Routledge, 2017.

\bibitem[Garcia-Erill and Albrechtsen(2020)]{garcia2020evaluation}
Gen{\'\i}s Garcia-Erill and Anders Albrechtsen.
\newblock Evaluation of model fit of inferred admixture proportions.
\newblock \emph{Molecular ecology resources}, 20\penalty0 (4):\penalty0 936--949, 2020.

\bibitem[Heinzel et~al.(2025)Heinzel, Baumdicker, and Pfaffelhuber]{heinzel2025revealing}
Carola~Sophia Heinzel, Franz Baumdicker, and Peter Pfaffelhuber.
\newblock Revealing the range of equally likely estimates in the admixture model.
\newblock \emph{G3: Genes, Genomes, Genetics}, page jkaf142, 2025.

\bibitem[Hoadley(1971)]{hoadley1971}
Bruce Hoadley.
\newblock Asymptotic properties of maximum likelihood estimators for the independent not identically distributed case.
\newblock \emph{The Annals of mathematical statistics}, pages 1977--1991, 1971.

\bibitem[Jakobsson and Rosenberg(2007)]{jakobsson2007}
Mattias Jakobsson and Noah~A Rosenberg.
\newblock Clumpp: a cluster matching and permutation program for dealing with label switching and multimodality in analysis of population structure.
\newblock \emph{Bioinformatics}, 23\penalty0 (14):\penalty0 1801--1806, 2007.

\bibitem[Kidd et~al.(2014)Kidd, Speed, Pakstis, Furtado, Fang, Madbouly, Maiers, Middha, Friedlaender, and Kidd]{kidd2014}
Kenneth~K Kidd, William~C Speed, Andrew~J Pakstis, Manohar~R Furtado, Rixun Fang, Abeer Madbouly, Martin Maiers, Mridu Middha, Fran{\c{c}}oise~R Friedlaender, and Judith~R Kidd.
\newblock Progress toward an efficient panel of snps for ancestry inference.
\newblock \emph{Forensic Science International: Genetics}, 10:\penalty0 23--32, 2014.

\bibitem[Kopelman et~al.(2015)Kopelman, Mayzel, Jakobsson, Rosenberg, and Mayrose]{kopelman2015}
Naama~M Kopelman, Jonathan Mayzel, Mattias Jakobsson, Noah~A Rosenberg, and Itay Mayrose.
\newblock Clumpak: a program for identifying clustering modes and packaging population structure inferences across k.
\newblock \emph{Molecular ecology resources}, 15\penalty0 (5):\penalty0 1179--1191, 2015.

\bibitem[Kosoy et~al.(2009)Kosoy, Nassir, Tian, White, Butler, Silva, Kittles, Alarcon-Riquelme, Gregersen, Belmont, et~al.]{kosoy2009ancestry}
Roman Kosoy, Rami Nassir, Chao Tian, Phoebe~A White, Lesley~M Butler, Gabriel Silva, Rick Kittles, Marta~E Alarcon-Riquelme, Peter~K Gregersen, John~W Belmont, et~al.
\newblock Ancestry informative marker sets for determining continental origin and admixture proportions in common populations in america.
\newblock \emph{Human mutation}, 30\penalty0 (1):\penalty0 69--78, 2009.

\bibitem[Lawson et~al.(2018)Lawson, Van~Dorp, and Falush]{lawson2018}
Daniel~J Lawson, Lucy Van~Dorp, and Daniel Falush.
\newblock A tutorial on how not to over-interpret structure and admixture bar plots.
\newblock \emph{Nature communications}, 9\penalty0 (1):\penalty0 3258, 2018.

\bibitem[Liu et~al.(2024)Liu, Kopelman, and Rosenberg]{liu2024clumppling}
Xiran Liu, Naama~M Kopelman, and Noah~A Rosenberg.
\newblock Clumppling: cluster matching and permutation program with integer linear programming.
\newblock \emph{Bioinformatics}, 40\penalty0 (1):\penalty0 btad751, 2024.

\bibitem[Lovell et~al.(2021)Lovell, MacQueen, Mamidi, Bonnette, Jenkins, Napier, Sreedasyam, Healey, Session, Shu, et~al.]{lovell2021}
John~T Lovell, Alice~H MacQueen, Sujan Mamidi, Jason Bonnette, Jerry Jenkins, Joseph~D Napier, Avinash Sreedasyam, Adam Healey, Adam Session, Shengqiang Shu, et~al.
\newblock Genomic mechanisms of climate adaptation in polyploid bioenergy switchgrass.
\newblock \emph{Nature}, 590\penalty0 (7846):\penalty0 438--444, 2021.

\bibitem[Pfaff et~al.(2004)Pfaff, Barnholtz-Sloan, Wagner, and Long]{pfaff2004information}
Carrie~Lynn Pfaff, Jill Barnholtz-Sloan, Jennifer~K Wagner, and Jeffrey~C Long.
\newblock Information on ancestry from genetic markers.
\newblock \emph{Genetic Epidemiology: The Official Publication of the International Genetic Epidemiology Society}, 26\penalty0 (4):\penalty0 305--315, 2004.

\bibitem[Pfaffelhuber and Rohde(2022)]{pfaffelhuber2022central}
Peter Pfaffelhuber and Angelika Rohde.
\newblock A central limit theorem concerning uncertainty in estimates of individual admixture.
\newblock \emph{Theoretical Population Biology}, 148:\penalty0 28--39, 2022.

\bibitem[Pfaffelhuber et~al.(2020)Pfaffelhuber, Grundner-Culemann, Lipphardt, and Baumdicker]{pfaffelhuber2020}
Peter Pfaffelhuber, Franziska Grundner-Culemann, Veronika Lipphardt, and Franz Baumdicker.
\newblock How to choose sets of ancestry informative markers: A supervised feature selection approach.
\newblock \emph{Forensic Science International: Genetics}, 46:\penalty0 102259, 2020.

\bibitem[Phillips et~al.(2019)Phillips, McNevin, Kidd, Lagac{\'e}, Wootton, De~La~Puente, Freire-Aradas, Mosquera-Miguel, Eduardoff, Gross, et~al.]{phillips2019maplex}
C~Phillips, D~McNevin, KK~Kidd, R~Lagac{\'e}, S~Wootton, M~De~La~Puente, A~Freire-Aradas, A~Mosquera-Miguel, M~Eduardoff, T~Gross, et~al.
\newblock Maplex-a massively parallel sequencing ancestry analysis multiplex for asia-pacific populations.
\newblock \emph{Forensic Science International: Genetics}, 42:\penalty0 213--226, 2019.

\bibitem[Phillips(2015)]{phillips2015}
Chris Phillips.
\newblock Forensic genetic analysis of bio-geographical ancestry.
\newblock \emph{Forensic Science International: Genetics}, 18:\penalty0 49--65, 2015.

\bibitem[Phillips et~al.(2014)Phillips, Parson, Lundsberg, Santos, Freire-Aradas, Torres, Eduardoff, B{\o}rsting, Johansen, Fondevila, et~al.]{phillips2014}
Christopher Phillips, W~Parson, B~Lundsberg, C~Santos, A~Freire-Aradas, M~Torres, M~Eduardoff, C~B{\o}rsting, P~Johansen, M~Fondevila, et~al.
\newblock Building a forensic ancestry panel from the ground up: The euroforgen global aim-snp set.
\newblock \emph{Forensic Science International: Genetics}, 11:\penalty0 13--25, 2014.

\bibitem[Pritchard and Donnelly(2001)]{pritchard2001}
Jonathan~K Pritchard and Peter Donnelly.
\newblock Case--control studies of association in structured or admixed populations.
\newblock \emph{Theoretical population biology}, 60\penalty0 (3):\penalty0 227--237, 2001.

\bibitem[Pritchard et~al.(2010)Pritchard, Wen, and Falush]{pritchard2010}
Jonathan~K Pritchard, Xiaoquan Wen, and Daniel Falush.
\newblock Documentation for structure software: Version 2.3.
\newblock \emph{University of Chicago, Chicago, IL}, pages 1--37, 2010.

\bibitem[Redner(1981)]{redner1981note}
Richard Redner.
\newblock Note on the consistency of the maximum likelihood estimate for nonidentifiable distributions.
\newblock \emph{The Annals of Statistics}, pages 225--228, 1981.

\bibitem[Resutik et~al.(2023)Resutik, Aeschbacher, Kr{\"u}tzen, Kratzer, Haas, Phillips, and Arora]{resutik2023}
Peter Resutik, Simon Aeschbacher, Michael Kr{\"u}tzen, Adelgunde Kratzer, Cordula Haas, Christopher Phillips, and Natasha Arora.
\newblock Comparative evaluation of the maplex, precision id ancestry panel, and visage basic tool for biogeographical ancestry inference.
\newblock \emph{Forensic Science International: Genetics}, 64:\penalty0 102850, 2023.

\bibitem[Rosenberg et~al.(2002)Rosenberg, Pritchard, Weber, Cann, Kidd, Zhivotovsky, and Feldman]{rb2002}
Noah~A Rosenberg, Jonathan~K Pritchard, James~L Weber, Howard~M Cann, Kenneth~K Kidd, Lev~A Zhivotovsky, and Marcus~W Feldman.
\newblock Genetic structure of human populations.
\newblock \emph{science}, 298\penalty0 (5602):\penalty0 2381--2385, 2002.

\bibitem[{The 1000 Genomes Project Consortium}(2015)]{10002015global}
{The 1000 Genomes Project Consortium}.
\newblock A global reference for human genetic variation.
\newblock \emph{Nature}, 526\penalty0 (7571):\penalty0 68--74, 2015.

\bibitem[Verdu et~al.(2014)Verdu, Pemberton, Laurent, Kemp, Gonzalez-Oliver, Gorodezky, Hughes, Shattuck, Petzelt, Mitchell, et~al.]{verdu2014}
Paul Verdu, Trevor~J Pemberton, Romain Laurent, Brian~M Kemp, Angelica Gonzalez-Oliver, Clara Gorodezky, Cris~E Hughes, Milena~R Shattuck, Barbara Petzelt, Joycelynn Mitchell, et~al.
\newblock Patterns of admixture and population structure in native populations of northwest north america.
\newblock \emph{PLoS genetics}, 10\penalty0 (8):\penalty0 e1004530, 2014.

\bibitem[Xavier et~al.(2020)Xavier, de~la Puente, Mosquera-Miguel, Freire-Aradas, Kalamara, Vidaki, Gross, Revoir, Po{\'s}piech, Kartasi{\'n}ska, et~al.]{xavier2020development}
Catarina Xavier, Maria de~la Puente, Ana Mosquera-Miguel, Ana Freire-Aradas, Vivian Kalamara, Athina Vidaki, Theresa~E Gross, Andrew Revoir, Ewelina Po{\'s}piech, Ewa Kartasi{\'n}ska, et~al.
\newblock Development and validation of the visage ampliseq basic tool to predict appearance and ancestry from dna.
\newblock \emph{Forensic Science International: Genetics}, 48:\penalty0 102336, 2020.

\bibitem[Xavier et~al.(2022)Xavier, de~la Puente, Sidstedt, Junker, Minawi, Unterl{\"a}nder, Chantrel, Laurent, Delest, Hohoff, et~al.]{xavier2022}
Catarina Xavier, Maria de~la Puente, Maja Sidstedt, Klara Junker, Angelika Minawi, Martina Unterl{\"a}nder, Yann Chantrel, Fran{\c{c}}ois-Xavier Laurent, Anna Delest, Carsten Hohoff, et~al.
\newblock Evaluation of the visage basic tool for appearance and ancestry inference using forenseq{\textregistered} chemistry on the miseq fgx{\textregistered} system.
\newblock \emph{Forensic Science International: Genetics}, 58:\penalty0 102675, 2022.

\end{thebibliography}
\end{small}

\appendix

\end{document}